\documentclass{amsart}

\RequirePackage[OT1]{fontenc}
\RequirePackage{amsthm,amsmath,amsfonts}
\RequirePackage[numbers]{natbib}
\RequirePackage[colorlinks,citecolor=blue,urlcolor=blue]{hyperref}
\usepackage{graphicx}
\usepackage{enumerate}

\numberwithin{equation}{section}
\theoremstyle{plain}

\newcommand{\R}{\mathbb{R}}
\newcommand{\union}{\cup}
\newcommand{\intersect}{\cap}

\newcommand{\given}{\ensuremath{\;|\;}}
\renewcommand{\vec}[1]{\boldsymbol{#1}}
\newcommand{\treespace}[1]{\ensuremath{\mathcal{T}_{#1}}}
\newcommand{\treespaceint}[1]{\ensuremath{\mathcal{T}_{#1}^{\text{int}}}}
\newcommand{\Xcodim}[1]{\ensuremath{X^{(#1)}}}

\newcommand{\paths}[1]{\ensuremath{C_{x_0,t_0}(#1)}}
\newcommand{\brownian}[1]{\ensuremath{B_{x_0,t_0}(#1)}}
\newcommand{\walk}[2]{\ensuremath{W_{x_0,t_0}^{#1}(#2)}}

\newcommand{\drop}{\ensuremath{\mathcal{P}}}

\newcommand{\borelX}{\ensuremath{{\widehat{\mathcal{X}}}}}
\newcommand{\W}[3]{\ensuremath{\mathcal{W}_{#1,#2}^{#3}}}
\newcommand{\B}[2]{\ensuremath{\mathcal{B}_{#1,#2}}}

\newcommand{\edgelen}[2]{\ensuremath{\ell_{#1}(#2)}}

\newcommand{\pr}[2][]{\ensuremath{\mathrm{Pr}_{#1}\!\left(#2\right)}}

\newcounter{theorem}
\newtheorem{thm}[theorem]{Theorem}
\newtheorem{lem}[theorem]{Lemma}

\theoremstyle{definition}
\newtheorem{defn}[theorem]{Definition}
\newcounter{algorithm}
\newtheorem{alg}[algorithm]{Algorithm}

\begin{document}

\begin{abstract}
Cubical complexes are metric spaces constructed by gluing together unit cubes in an analogous way to the construction of simplicial complexes. 
We construct Brownian motion on such spaces, define random walks, and prove that the transition kernels of the random walks converge to that for Brownian motion. 
The proof involves pulling back onto the complex the distribution of Brownian sample paths on the standard cube, and combining this with a distribution on walks between cubes in the complex. 
The main application lies in analysing sets of evolutionary trees: several tree spaces are cubical complexes and we briefly describe our results and some applications in this context. 
Our results extend readily to a class of polyhedral complex in which every cell of maximal dimension is isometric to a given fixed polyhedron.  
\end{abstract}

\title[Random walks and {B}rownian motion on cubical complexes]{Convergence of random walks to {B}rownian motion on cubical complexes}

\author{Tom M.\ W.\ Nye}
%\address{School of Mathematics and Statistics, Newcastle University, UK}

\address{School of Mathematics, Statistics and Physics\\
Newcastle University\\UK}
\email{tom.nye@ncl.ac.uk}
\urladdr{http://www.mas.ncl.ac.uk/$\sim$ntmwn/}

\maketitle

\section{Introduction}

Most statistical analysis is carried out on data lying in a Euclidean vector space, or more generally, a smooth manifold. 
Manifold-stratified spaces, however, are playing an increasingly important role in applications and have recently attracted substantial research interest \cite{fer14b,fer14a,ooda14}. 
In a manifold-stratified space, each stratum is a manifold, and strata are glued together along their boundaries. 
In this way, the structure is smooth in the interior of strata, but can become singular where strata are joined together. 
Important examples include shape spaces \cite{kendall1999} and, of particular relevance to this article, spaces of trees \cite{bill01,fera11a,gav2016} and networks \cite{dev16}. 
Cubical complexes are manifold-stratified spaces built by joining high-dimensional cubes at their faces \cite{gromov1987,niblo1998}.  
They are polyhedral complexes in which every cell is a unit cube, and cubes can have different dimensions. 
Each cube is equipped with the $L^2$ Euclidean metric, and for certain cubical complexes this extends to define a metric on the whole complex. 
Cubical complexes provide a concrete way to construct stratified sample spaces, and various different spaces of evolutionary trees are constructed in this way \cite{bill01,gav2016,dev16}. 
The tree space of Billera, Holmes and Vogtmann \cite{bill01}, usually referred to as \emph{BHV tree space}, in particular has received considerable research attention, with methods for computing Fr\'echet means and variances \cite{mill15}, constructing principal components \cite{nye11}, and performing other statistical tests \cite{willis16}. 
In addition, asymptotic results for Fr\'echet means have also been established on BHV tree space \cite{barden2013central} and some other related spaces \cite{hotz13,barden2017}. 
However, beyond these specific asymptotic results, fundamental probability theory has not been adequately developed for general classes of stratified spaces. 

A key problem for carrying out statistics on stratified spaces is the difficulty of constructing tractable distributions which can be used for statistical inference. 
For example, consider the uniform distribution on the ball $B(x,r)$ with radius $r$ at a point $x$ in a cubical complex $X$. 
For fixed $r$, the volume of $B(x,r)$ can vary with $x$ since the ball potentially intersects a different number of cubes depending on the location of $x$. 
As a result, even very simple distributions can have normalizing constants which are computationally intractable, depending on the structure of $X$. 
An example of this problem in BHV tree space is described in more detail by Weyenberg et al \cite{wey16}.

On the other hand, tractable distributions could be constructed as the transition kernels of suitably defined stochastic processes on $X$, and used to build statistical models. 
For example, the transition kernel at time $t_0$ for a Brownian motion starting from $x_0\in X$ at time zero is a natural analog of a Euclidean multivariate normal distribution. 
Fitting this distribution to data in $X$ would give estimates of $x_0$ and $t_0$ which could be interpreted as a mean and variance in $X$.  
With this idea as motivation, this article considers Brownian motion and random walks on cubical complexes, and convergence of the transition kernels of random walks to the transition kernel for Brownian motion. 
Brownian motion on $X$ is defined in the following way. 
On the interior of each cube, it proceeds in the same way as in Euclidean space; at the boundary of a cube, the particle undergoing Brownian motion can potentially move into several adjacent cubes, depending on the structure of the complex, and does so uniformly at random. 
Specifically, if the particle is contained in the boundary of $r$ cubes of dimension $m$, and no cubes of dimension $>m$, then it moves with probability $1/r$ into each of the adjacent $m$-cubes. 
A simple example consists of the finite $3$-spider, the cubical complex formed by joining three copies of the interval $[0,1]$ at the shared origin $0$. 
A particle incident at the origin moves with equal probability onto each of the three adjacent line segments. 
The heat equation on the finite $3$-spider has an analytic solution \cite{nye14a}, but analytic solutions will not exist for general $X$. 
Random walks can be used to obtain a computational approximation to Brownian motion in order to perform inference for statistical models, but only if the random walk converges to Brownian motion in an appropriate sense. 
This is exactly our main result, Theorem~\ref{thm:conv}, in which we prove that the transition kernels of suitably defined random walks converge to Brownian motion on a general class of cubical complex. 
The method of proof involves projecting sample paths from the complex down onto a single cube and then using the pull-back to transfer measures and sets from Euclidean space back onto the complex. 
This enables an explicit construction of the transition kernel of the Brownian motion on $X$. 
The main technical challenge involved in proving Theorem~\ref{thm:conv} arises as Brownian sample paths can cross the boundaries between cubes an infinite number of times, and some careful analysis is required to handle this difficulty. 

The existing literature on Brownian motion and random walks in manifold-stratified spaces is limited. 
Frank and Durham \cite{frank1984} studied Brownian motion on trees with absorbing states at the leaves. 
Brin and Kifer~\cite{brin01} considered Brownian motion on $2$-dimensional Euclidean simplicial complexes. 
They proved existence of Brownian motion on these complexes and established properties of the infinitesimal generator. 
However they did not consider random walks and so the majority of their results are unrelated to this paper. 
Later, Enriquez and Kifer~\cite{enri01} proved a version of Donsker's theorem on metric graphs and showed that random walks along the edges of graphs converge to Brownian motion. 
The task of proving convergence of random walk to Brownian motion on graphs is a lower-dimensional analogue of proving convergence on cubical complexes. 
Indeed, they stated that they viewed their result as a first step towards considering random walks on more general simplicial complexes. 
More recently, Kostrykin et al \cite{kost2012} studied more general Brownian motions on graphs. 
There is some literature concerning random walks on groups which act on cubical complexes \cite{fernos2017}, but these random walks are not directly related to the geometrical random walks we consider here. 
In addition there is earlier work by Kendall~\cite{kendall1977}, concerning diffusion in certain quotient spaces. 

The remainder of the paper is structured as follows. 
Section~\ref{sec:prelim} contains background material on cubical complexes. 
In Section~\ref{sec:BMcubical} we construct Brownian motion on a cubical complex $X$ as a probability measure on paths in $X$ and show this corresponds to a  well-defined Markov process. 
In Section~\ref{sec:conv} we define random walks on $X$ and prove they converge to Brownian motion under a certain limit. 
We discuss applications of our results to BHV tree space in Section~\ref{sec:BHV}.
Finally, in Section~\ref{sec:concl}, we consider generalizations of Theorem~\ref{thm:conv}  and make some concluding remarks.

\section{Preliminaries}\label{sec:prelim}

\subsection{Cubical complexes}

Let $I=[0,1]$ be the unit interval and let $I^n$ be the $n$-dimensional cube, or \emph{$n$-cube}, equipped with the Euclidean metric. 
A codimension-$k$ face of $I^n$ is the subset of points for which $k$ coordinates are individually fixed to either $0$ or $1$. 
A cubical complex $X$ is a space formed by gluing together cubes (potentially of different dimensions) via isometric identification of various faces. 
%More formally, $X$ is a polyhedral complex in which every cell is a cube and all attaching maps are injective. 
Simple examples include graphs on which every edge is identified with the unit interval, and the positive orthant in $\R^n$, formed by filling the orthant with an infinite number of unit cubes $I^n$.
More interesting examples include spaces of evolutionary trees, as discussed in Section~\ref{sec:BHV}. 
We will assume that $X$ contains cubes that are at most $n$-dimensional, and that $X$ contains at least one cube of dimension $n$, for some fixed value of $n$.
We also assume that $X$ is locally finite, or in other words, that every point $x\in X$ is contained in a finite number of cubes.  
The Euclidean metric on each cube extends to the whole of $X$ in the following way. 
When $x,y\in X$ are in different cubes the distance $d(x,y)$ is defined to be the infimum of lengths of paths between $x$ and $y$ which are straight lines segments in each cube. 
The Hopf-Rinow theorem implies that $X$ is then a geodesic metric space, i.e.\ for all $x,y\in X$ the distance $d(x,y)$ is realized as the length of at least one shortest length path between $x$ and $y$ \cite[Proposition I.3.7]{BH1999}. 
We do not make any additional assumptions about this metric -- in particular we do not require $X$ to have non-positive curvature. 
The cubical complex $X$ is then equipped with the Borel sigma algebra defined using this metric. 
The restrictions of these Borel sets to any $n$-cube are exactly the Borel subsets of $I^n$.  

We will primarily be concerned with $n$-cubes in $X$, where $n$ is the highest dimension in $X$, and how they are attached to each other via codimension-$1$ faces. 
This information can be encapsulated in a bipartite graph $G_X=(U,V,E)$ where the vertices $U$ correspond to $n$-cubes in $X$; the vertices $V$ correspond to $(n-1)$-cubes in $X$; and $(u,v)\in E\subseteq U\times V$ is an edge if and only if $v$ is a codimension-$1$ face of $u$. 
We assume that $G_X$ is connected. 
Brownian motion and random walks will start from a fixed point which we assume lies in the interior of an $n$-cube $u_0\in U$, unless stated otherwise. 
Sample paths of Brownian motion on $X$ will be associated with certain walks on $G_X$, and so we make the following definition. 

\begin{defn}
A $k$-walk on $G_X$ is a sequence $u_0,v_1,u_1,v_2,\ldots,u_{k-1},v_k,u_k$ 
where $u_0,\ldots,u_k\in U$ and $v_1,\ldots,v_k\in V$, and for which every pair of adjacent vertices in the sequence corresponds to an edge so that $(u_{i-1},v_{i})\in E$ and $(v_i,u_i)\in E$ for $i=1,\ldots,k$. 
The walk is \emph{non-repeating} if $v_i\neq v_{i+1}$ for all $i=1,\ldots,k-1$. 
The set of all non-repeating $k$-walks is denoted $W_k$. 
\end{defn}

\subsection{Brownian motion on $I^n$}\label{sec:bm-cube}

When $X=I^n$ we consider a reflected Brownian motion on the cube: each coordinate on $I^n$ evolves independently according to a standard Brownian motion on $[0,1]$ with reflecting boundaries at $0$ and $1$. 
The corresponding distribution on $\paths{I^n}$, the set of paths starting at $x_0\in I^n$ and defined on the interval $[0,t_0]$, is denoted $\brownian{I^n}$. 
Let $F$ denote the set of codimension-$1$ faces of $I^n$. 
Given a path $\eta\in\paths{I^n}$ let $\pi(\eta)$ denote the time-ordered sequence of faces $f_1,f_2\ldots\in F$ that $\eta$ meets, with repeated consecutive intersections with the same face $f$ represented by a single occurrence of $f$ in $\pi(\eta)$, so that $f_i\neq f_{i+1}$ for all $i=1,2,\ldots$. 
Every time a path drawn from $\brownian{I^n}$ meets a face $f$, is does so almost surely an infinite number of times before meeting a different face $f'\neq f$. 
It is these repetitions which are removed to form the sequence $\pi(\eta)$, so that each element in the sequence potentially represents an infinite number of intersections with that face.  

\begin{lem}\label{lem:facehits}
Suppose $x_0\in I^n$ does not lie in a codimension-$2$ face.  
If $\eta$ is drawn from $\brownian{I^n}$ then $\pi(\eta)$ is almost surely well-defined and finite. 
\end{lem}

\begin{proof}
Brownian motion in the plane hits the origin with probability zero, and it follows that sample paths on $I^n$ with $\eta(0)=x_0$ avoid codimension-$2$ faces almost surely. 
The sequence $\pi(\eta)$ is therefore almost surely well-defined. 
Let $A_k\subset \paths{I^n}$ be the set of paths such that $\pi(\eta)$ has length $k$, for $k=0,1,\ldots$ and let $A_\infty$ be the complement of $\bigcup_k A_k$. 
We claim that $A_\infty$ has measure zero. 
Let $B_\rho\subset \paths{I^n}$ be the set of paths bound a distance $\rho$ away from all codimension-2 faces of $I_n$. 
Suppose $\eta$ is drawn from $\brownian{I^n}$ and $\eta\in A_\infty\cap B_\rho$. 
Then there are two faces $f\neq f'$ of $I^n$ which $\eta$ alternately hits an infinite number of times. 
It follows that there exists an increasing sequence of times $T_0,T_1,\ldots\in[0,t_0]$ such that $\eta(T_{2r})\in f$ and $\eta(T_{2r+1})\in f'$ for all $r=0,1,\ldots$. 
Since $\eta$ is bounded by distance $\rho$ away from codimension-2, there is a coordinate $i$ such that
\begin{equation*}
|\eta_i(T_{2r})-\eta_i(T_{2r+1})|\geq \rho
\end{equation*}
for all $r$. 
%Without loss of generality, we can assume $\eta_i(T_{2r})=0$ so that the inequality becomes $\eta_i(T_{2r+1})\geq\rho$. 
By definition, $\eta_i(t)$ is a reflected Brownian motion on $[0,1]$. 
The probability of the event $\eta\in A_\infty\cap B_\rho$ is therefore bounded above by the probability that the reflected Brownian motion $\eta_i$ moves a distance at least $\rho$ over the time interval $T_{2r+1}-T_{2r}$:
\begin{equation*}
\brownian{I^n}(A_\infty\cap B_\rho)\leq \pr{Z(T_{2r+1}-T_{2r})\geq \rho}
\end{equation*}
where $Z(t)$ is a reflected Brownian motion on $[0,1]$ with $Z(0)=0$. 
Since $T_{2r+1}-T_{2r}\rightarrow 0$ as $r\rightarrow\infty$, the right-hand side of the inequality can be made arbitrarily small for fixed $\rho$, and so $A_\infty\cap B_\rho$ has zero measure. 
The claim follows by taking the limit as $\rho\rightarrow 0$.
\end{proof}

\section{Brownian motion on cubical complexes}\label{sec:BMcubical}

In this section we give an explicit construction of Brownian motion on the cubical complex $X$. 
The construction relies upon a projection map which takes paths in $X$ to paths on $I^n$. 
The projection enables us to write down a probability measure on sets of paths in $X$ in terms of the reflected Brownian motion on $I^n$. %, and this is taken as the definition of the Brownian motion on $X$. 

Fix an $n$-cube $u_0\in U$ and a point $x_0$ in the interior of $u_0$. 
In analogy to the notation in Section~\ref{sec:bm-cube}, we let $\paths{X}$ be the metric space of continuous maps $\eta:[0,t_0]\rightarrow X$ which satisfy $\eta(0)=x_0$.
The metric is defined by 
\begin{equation*}
d_\infty(\eta_1,\eta_2) = \sup_{t\in[0,t_0]} d (\eta_1(t),\eta_2(t))
\end{equation*}
for any two $\eta_1,\eta_2\in\paths{X}$ where $d(\cdot,\cdot)$ denotes the metric on $X$. 
Let $X^{(k)}\subset X$ denote the union of codimension-$k$ faces in $X$. 

Define $\widehat{C}$ by
\begin{equation*}
\widehat{C} = \{\eta \in \paths{X}: \text{im}(\eta)\intersect\Xcodim{2}=\emptyset
\text{\ and\ } \eta(t_0)\in X\setminus\Xcodim{1}\},
\end{equation*}
where $\text{im}(\eta)$ denotes the image of $\eta$. 
%Every element of $\widehat{C}$ determines a walk on the graph $G_X$, possibly of infinite length. 
For any path $\eta\in \widehat{C}$ we define a corresponding path $\drop(\eta)\in\paths{u_0}$ as follows.  
We define $\drop(\eta)$ to be the same as $\eta$ until $\eta$ first crosses a codimension-$1$ face $v_1$ of $u_0$ into a different $n$-cube $u_1\in U$. 
The isometry gluing $u_0$ to $u_1$ via $v_1$ extends uniquely to determine an isometry between $u_1$ and $u_0$. 
This isometry is used to define the continuation of $\drop(\eta)$ in $u_0$ until $\eta$ hits $\Xcodim{1}$ again, at which point $\drop(\eta)$ simultaneously hits a face of $u_0$. 
The process is repeated for all $t\in[0,t_0]$ to define a continuous map $\drop: \widehat{C} \rightarrow\paths{u_0}$. 
It is convenient to use a fixed isometry between $u_0$ and $I^n$, so that the map $\drop$ in fact takes values in $\paths{I^n}$. 
An example of paths on $X$ and their projections is shown in Figure~\ref{fig:projection}. 

\begin{figure}
\begin{center}
\includegraphics[scale=0.65]{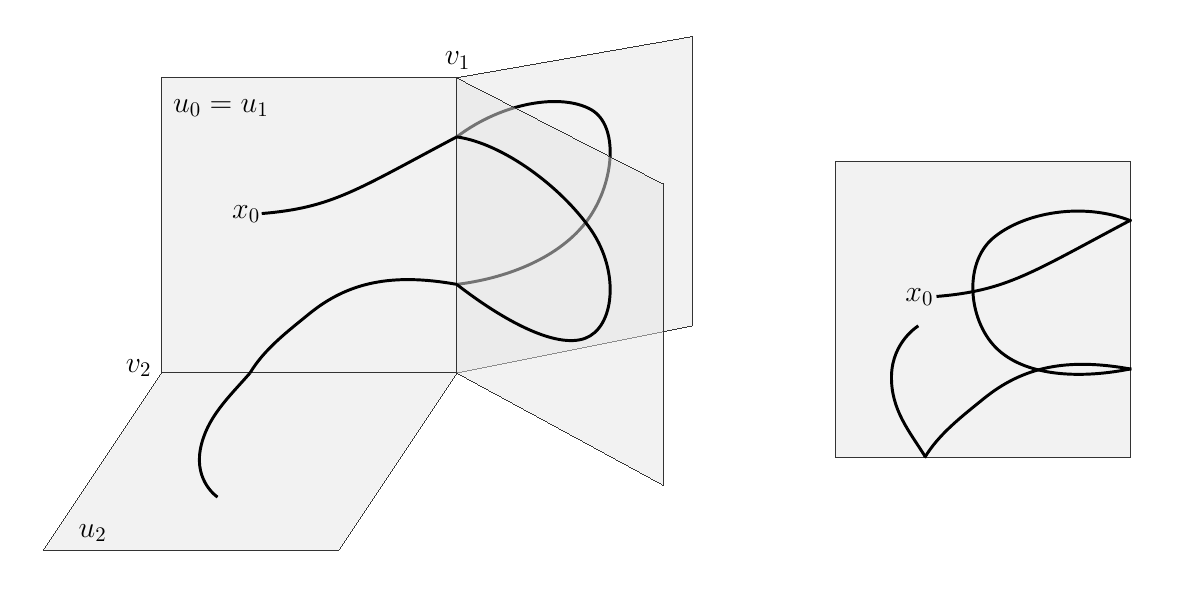}
\caption{Illustration of the projection map for a cubical complex $X$ consisting of four copies of $I^2$, glued together as shown. Left: two paths $\eta_1,\eta_2$ in $X$, constructed so that $\drop(\eta_1)=\drop(\eta_2)$. Right: the corresponding projections $\drop(\eta_i)$ onto $I^2$. 
Both paths $\eta_1,\eta_2$ lie in the same set $\widehat{C}_{\vec{f}}^{\vec{w}}$ with $k=2$ but traverse a different set of $2$-cubes. 
The associated non-repeating $G_X$ walk is $u_0,v_1,u_1,v_2,u_2$.  
}
\label{fig:projection}
\end{center}
\end{figure}

We say a sequence $(f_1,\ldots,f_k)$ of elements in $F$ is \emph{non-repeating} if $f_i\neq f_{i+1}$ for all $i=1,\ldots,k-1$. 
The set of non-repeating sequences of length $k$ is denoted $F_k$. 
For any such sequence, define ${C}_{f_1,\ldots,f_k}$ by
\begin{multline*}
{C}_{f_1,\ldots,f_k} = \{ \eta\in\paths{I^n} : \pi(\eta)\ \text{is well defined}\\ \text{with\ }\pi(\eta)=(f_1,\ldots,f_k)\ \text{and\ }\eta(t_0)\in\text{int}(I^n) \}.   
\end{multline*}
This definition includes for $k=0$ the set $C_{\ast}$ of paths which do not meet any face.  
For brevity we will usually write $\vec{f}=(f_1,\ldots,f_k)$ for an element of $F_k$. 
Lemma~\ref{lem:facehits} shows that
\begin{equation}\label{equ:total_prob}
\sum_{k\geq 0}\sum_{\vec{f}\in F_k}\brownian{I^n}(C_{\vec{f}})=1.
\end{equation} 
We need to consider the pre-images
\begin{equation*}
\widehat{C}_{\vec{f}} = \drop^{-1}(C_{\vec{f}}).
\end{equation*}
Suppose $\eta\in \widehat{C}_{\vec{f}}$ is the pre-image of a Brownian sample path on $I^n$, and consider the walk on $G_X$ determined by $\eta$.
Since every time $\drop(\eta)$ hits a face $f\in F$ it does so almost surely an infinite number of times, it follows that $\eta$ determines an infinite walk on $G_X$ almost surely whenever $\vec{f}$ is not empty. 
However, if we remove repeated crossings of each facet $v$ from the walk, in an analogous way to the definition of $\pi(\eta)$ in Section~\ref{sec:bm-cube}, we obtain a non-repeating $k$-walk on $G_X$. 
Specifically, if $v_1,\ldots,v_k$ is the list of distinct codimension-$1$ faces $\eta$ crosses, and $u_{i-1}$ is the last $n$-cube reached before $\eta$ meets $v_i$, $i=1,\ldots,k-1$, then the non-repetitive walk associated with $\eta$ is $u_0,v_1,u_1,v_2,\ldots,u_{k-1},v_k,u_k$. 
The $k$-walk associated with $\eta$ is denoted $\phi(\eta)$, and we let $W(\vec{f}) = \{ \phi(\eta):\eta\in \widehat{C}_{\vec{f}}\}\subseteq W_k$ when $\vec{f}\in F_k$. 
For any $\vec{w}\in W(\vec{f})$, we let $\widehat{C}_{\vec{f}}^{\vec{w}}=\{\eta\in\widehat{C}_{\vec{f}}:\phi(\eta)=\vec{w}\}$, 
% denote the set paths corresponding to walk $\vec{w}$, 
so that 
\begin{equation*}
\widehat{C}_{\vec{f}} = \bigcup_{\vec{w}\in W(\vec{f})}\widehat{C}_{\vec{f}}^{\vec{w}}.
\end{equation*}
For fixed $\vec{f}$, each set of paths $\widehat{C}_{\vec{f}}^{\vec{w}}$ can be thought of as a different way of pulling-back the paths $C_{\vec{f}}$ onto $X$ by choosing different adjacent cubes when the paths finish crossing each codimension-$1$ face in $X$. 

Given any $\vec{w}=(u_0,v_1,u_1,v_2,\ldots,u_{k-1},v_k,u_k)\in W_k$, let
\begin{equation}\label{equ:pw}
p(\vec{w}) = \prod_{i=1}^k \left( \text{deg}\left( v_i \right) \right)^{-1}.
\end{equation}
Then, induction on $k$ shows that for each $\vec{f}\in F_k$,
\begin{equation*}
\sum_{\vec{w}\in W(\vec{f})} p(\vec{w}) = 1,
\end{equation*}
so the values $\{p(\vec{w}):\vec{w}\in W(\vec{f})\}$ determine a probability distribution on $W(\vec{f})$. 
%This distribution corresponds to choosing uniformly at random which adjacent cube to move into when a path hits a codimension-$1$ face in $X$, when pulling-back paths from $C_{\vec{f}}$ onto $X$. 

We will be particularly concerned with the set of Brownian sample paths which end in some region of $U\subseteq X$. 
For any $k$, $\vec{f}\in F_k$ and $\vec{w}\in W(\vec{f})$, let $\widehat{C}_{\vec{f}}^{\vec{w}}(U)$ be the subset of $\widehat{C}_{\vec{f}}^{\vec{w}}$ consisting of paths $\eta$ with $\eta(t_0)\in U$. 
This set will be empty unless the walk $\vec{w}$ ends in a cube which intersects $U$. 
Similarly, let
\begin{equation*}
\widehat{C}_{\vec{f}}(U) = \bigcup_{\vec{w}\in W(\vec{f})}\widehat{C}_{\vec{f}}^{\vec{w}}(U).
\end{equation*}
Let $\borelX$ be the $\sigma$-algebra on $\widehat{C}$ generated by sets of the form $\widehat{C}_{\vec{f}}^{\vec{w}}(U)$ where $U\subseteq X$ is a Borel set. 

\begin{lem}\label{lem:borel_sets}
$\borelX$ is a sub-algebra of the Borel $\sigma$-algebra on $\widehat{C}$. 
Furthermore, each set in $\borelX$ is mapped to a Borel subset of $\paths{I^n}$ by $\drop$.  
\end{lem}

\begin{proof}
An inductive argument can be used to prove that $C_{\vec{f}}$ is a Borel set as follows. 
A sufficiently small perturbation of any $\eta\in C_{\vec{f}}$ gives either a path with $\pi(\eta)=\vec{f}$ or for which $\pi(\eta)$ is a subsequence of $f$. 
Thus the union of $C_{\vec{f}}$ and all sets $C_{\vec{f}'}$ where $\vec{f}'$ is a subsequence of $\vec{f}$ is open. 
By induction, $C_{\vec{f}}$ is therefore the intersection of an open set and a Borel set, and so is a Borel set. 
A similar argument applies to $\widehat{C}_{\vec{f}}^{\vec{w}}$. 
If $U\subseteq X$ is Borel, then subset of $\widehat{C}$ consisting of paths which end in $U$ is Borel, and hence so is $\widehat{C}_{\vec{f}}^{\vec{w}}(U)$. 
It follows that $\borelX$ is generated by a subset of the Borel $\sigma$-algebra on $\widehat{C}$, and so $\borelX$ is a sub-algebra. 
For fixed $\vec{f}$ and $\vec{w}$ the set $U\subseteq X$ determines a subset $V\subseteq I^n$ corresponding to the projection of the end points of paths in $\widehat{C}_{\vec{f}}^{\vec{w}}(U)$ under $\drop$. 
If $U$ is Borel, then $V$ is too, and $\drop\left( \widehat{C}_{\vec{f}}^{\vec{w}}(U) \right)$ is the set of paths in $C_{\vec{f}}$ which end in $V$. 
This is a Borel subset of $\paths{I^n}$. 
Similarly
\begin{equation*}
\drop\left( \left( \bigcup_k\bigcup_{\vec{f}\in F_K} \bigcup_{\vec{w}\in W(\vec{f})} \widehat{C}_{\vec{f}}^{\vec{w}} \right)^c\,\,\right) = \left( \bigcup_k\bigcup_{\vec{f}\in F_K} C_{\vec{f}}  \right)^c
\end{equation*}
and this is also a Borel subset of $\paths{I^n}$. 
It follows that every set in $\borelX$ maps under $\drop$ to a Borel subset of $\paths{I^n}$ since the image of any such set can be expressed via countable unions and complements of the above sets. 
\end{proof}

Given $\eta\in\widehat{C}_{\vec{f}}$, the set $\drop^{-1}\drop(\eta)$ is a symmetrization of $\eta$ around codimension-$1$ faces in $X$. 
This leads to the following technical lemma which is required later. 

\begin{lem}\label{lem:symm}
For all $U\subseteq X$ and all $\vec{f}\in F_k, \vec{w}\in W(\vec{f})$
\begin{equation}\label{equ:symmclosed}
 \widehat{C}_{\vec{f}}^{\vec{w}} \intersect \drop^{-1}\drop \left( \widehat{C}_{\vec{f}}^{\vec{w}}(U) \right) = \widehat{C}_{\vec{f}}^{\vec{w}}(U)
\end{equation}
\end{lem} 

\begin{proof}
It is clear that $\drop^{-1}\drop ( \widehat{C}_{\vec{f}}^{\vec{w}}(U) )$ contains $\widehat{C}_{\vec{f}}^{\vec{w}}(U)$.
%, and $\drop^{-1}\drop ( \widehat{C}_{\vec{f}}^{\vec{w}}(U) )\subset\widehat{C}_{\vec{f}}$ since $\drop ( \widehat{C}_{\vec{f}}^{\vec{w}}(U) )\subset C_{\vec{f}}$. 
It remains to show the left hand side of~$\eqref{equ:symmclosed}$ contains no elements outside $\widehat{C}_{\vec{f}}^{\vec{w}}(U)$. 
So suppose $\eta\in\widehat{C}_{\vec{f}}^{\vec{w}}(U)$ and $\eta'\in\widehat{C}_{\vec{f}}^{\vec{w}}$ satisfy $\drop(\eta)=\drop(\eta')$, so that $\eta'$ is a general element of the left hand side of~$\eqref{equ:symmclosed}$. 
%Then $\eta,\eta'$ are both contained in the left hand side of~$\eqref{equ:symmclosed}$. 
Since both paths are contained in $\widehat{C}_{\vec{f}}^{\vec{w}}$, it follows that $\eta(t_0)$ and $\eta'(t_0)$ lie in the same $n$-cube and also enter that $n$-cube from the same codimension-1 face. 
Then $\drop(\eta)=\drop(\eta')$ implies that $\eta$ and $\eta'$ are identical in the cube containing the paths at time $t_0$, and in particular $\eta(t_0)=\eta'(t_0)$ so $\eta'\in\widehat{C}_{\vec{f}}^{\vec{w}}(U)$. 
This establishes equation~$\eqref{equ:symmclosed}$.
Note that it is not necessarily the case that $\eta=\eta'$: the paths can pass through different $n$-cubes whenever $\drop(\eta)$ hits the same face of $I^n$ multiple consecutive times -- see Figure~\ref{fig:projection}. 
\end{proof}

We can now define a measure on $(\widehat{C}, \borelX)$ corresponding to the distribution of sample paths under Brownian motion on $X$. 

\begin{defn}\label{def:brownianX}
Define the probability measure $\brownian{X}$ on $(\widehat{C}, \borelX)$ by
\begin{equation}\label{equ:BrownianX}
\brownian{X}(A) = \sum_{k=0}^\infty \sum_{\vec{f}\in F_k}\sum_{\vec{w}\in W(\vec{f})} p(\vec{w}) \brownian{I^n}\left(\drop(A\intersect\widehat{C}_{\vec{f}}^{\vec{w}})\right)
\end{equation}
for any $A\in\borelX$. 
\end{defn}

\begin{lem}\label{lem:exists}
Equation~$\eqref{equ:BrownianX}$ defines a probability measure on $(\widehat{C}, \borelX)$. 
\end{lem}

\begin{proof}
Each term in the sum is defined since, for all $A\in\borelX$, the set $\drop\left( A\intersect \widehat{C}_{\vec{f}}^{\vec{w}} \right)$ is a Borel subset of $\paths{I^n}$ and so $\brownian{I^n}\left(\drop(A\intersect\widehat{C}_{\vec{f}}^{\vec{w}})\right)$ is well-defined. 
Next, the total measure is given by
\begin{align*}
\brownian{X}(X) &=
\brownian{X}\left( \bigcup_k\bigcup_{\vec{f}\in F_K} \bigcup_{\vec{w}\in W(\vec{f})} \widehat{C}_{\vec{f}}^{\vec{w}}\right)\\ 
&= \sum_{k=0}^\infty \sum_{\vec{f}\in F_k} \sum_{\vec{w}\in W(\vec{f})}p(\vec{w}) \brownian{I^n}\left(\drop(\widehat{C}_{\vec{f}}^{\vec{w}})\right).
\end{align*}
But $\drop(\widehat{C}_{\vec{f}}^{\vec{w}})=C_{\vec{f}}$ so
\begin{align*}
\brownian{X}(X) &= 
\sum_{k=0}^\infty \sum_{\vec{f}\in F_k} \brownian{I^n}\left( C_{\vec{f}} \right) 
\sum_{\vec{w}\in W(\vec{f})}p(\vec{w})\\
&= \sum_{k=0}^\infty \sum_{\vec{f}\in F_k} \brownian{I^n}\left( C_{\vec{f}} \right)\\
&=1
\end{align*}
using equation~$\eqref{equ:total_prob}$.  
Finally, suppose $A_1,A_2,\ldots\in\borelX$ are disjoint sets.
We aim to show countable additivity. 
Without loss of generality we can assume each set $A_l$ is contained within a single set of the form $\widehat{C}_{\vec{f}}^{\vec{w}}$, since if this is not the case, we can decompose the sets $A_l$ into a larger collection of disjoint sets. 
We also assume at most one set $A_l$ is contained in each set $\widehat{C}_{\vec{f}}^{\vec{w}}$, so that $\vec{f}_l$ and $\vec{w}_l$ are the sequences of faces and the non-repetitive walk associated with $A_l$.  
Then
\begin{align*}
\brownian{X}(\bigcup A_l)&=
\sum_{k=0}^\infty \sum_{\vec{f}\in F_k} \sum_{\vec{w}\in W(\vec{f})} p(\vec{w})\brownian{I^n}\left( \drop(\widehat{C}_{\vec{f}}^{\vec{w}}\intersect\bigcup A_l)\right)\\
&= \sum_l p(\vec{w}_l) \brownian{I^n}(\drop(A_l))\\
&= \sum_l \brownian{X}(A_l).
\end{align*}
If there is more than one set contained in each $\widehat{C}_{\vec{f}}^{\vec{w}}$ the result still holds, since, for example, if $A,B\subset \widehat{C}_{\vec{f}}^{\vec{w}}$ are disjoint then $\brownian{I^n}(\drop(A\union B)) = \brownian{I^n}(\drop(A))+\brownian{I^n}(\drop(B))$. 
\end{proof}

The measure on paths determined by equation~$\eqref{equ:BrownianX}$ corresponds to the diffusion of particles on the cubical complex $X$ defined in the Introduction. 
The projection map $\drop$ ensures that within each $n$-cube the paths of diffusing particles correspond to sample paths in Euclidean space. 
On the other hand, the distribution on the walks on $G_X$ determined by the probabilities $p(\vec{w})$ ensures that when a particle hits a codimension-$1$ boundary it moves with equal probability into each neighbouring $n$-cube. 
More specifically, using equation~$\eqref{equ:BrownianX}$, it is easy to show that if $\eta$ is sampled from $\brownian{X}$, then $\pr{\eta\in \widehat{C}_{\vec{f}}^{\vec{w}}\given \eta\in \widehat{C}_{\vec{f}}} = p(\vec{w})$. 
It follows that for $0<t<t_0$ the position $Z(t)\in X$ of a particle along a path sampled from $\brownian{X}$ is a well-defined Markov process. 
In fact, it is possible to write down an explicit formula for the transition kernel of the Markov process, as given in Definition~\ref{def:trans_kernel} later.

\section{Random walks and convergence}\label{sec:conv}

In this Section we define random walks on $X$ and prove they converge to Brownian motion in an appropriate sense. 
Random walk on $X$ is defined by the following algorithm. 
As input it takes $x_0\in X$ (assumed to be in the interior of some $n$-cube), $t_0>0$ and $m$, the number of steps of the random walk. 
We fix $\epsilon = (3t_0/m)^{1/2}$, so that the uniform distribution on $[-\epsilon,\epsilon]$ has variance $t_0/m$. 
The number of steps $m$ must be sufficiently large that $\epsilon<1$. 
We maintain a list $L$ throughout the algorithm, with $L$ initially empty. 
The list is used to record points at which the random walk crosses codimension-$1$ faces. 

\begin{alg}\label{alg:rw}
Let $y_0=x_0$.
For $j=1,2,\ldots,r$ where $r=m\times n$ repeat the following:
\begin{enumerate}
\item Let $u$ be an $n$-cube containing $y_{j-1}$ and let $\vec{\xi}:u\rightarrow I^n$ be an isometry, so that $\xi_i(y)$ denotes the $i$-th coordinate of any point $y\in u$ and $i=1,\ldots,n$. 
\item Sample $k$ uniformly at random from $1,\ldots,n$ and sample $\delta \sim U(-\epsilon, \epsilon)$. 
\item Perform one of the following steps, depending on the value of $\ell^* := \xi_k(y_{j-1})+\delta$.
\begin{enumerate}
\item If $0\leq \ell^* \leq 1$ then let $y_j$ be the point in $u$ with $\xi_k(y_j)=\ell^*$ and all other coordinates equal to $y_{j-1}$.
\item Otherwise either $\ell^*<0$ or $1<\ell^*$. 
Let $v$ be the $(n-1)$-cube associated with $\xi_k=0$ or $\xi_k=1$ respectively. 
Let $u^*$ be an $n$-cube sampled uniformly at random from the $n$-cubes which share the face $v$, including $u$ itself. 
There is then a unique isometry $\vec{\xi}^*:u^*\rightarrow I^n$ satisfying $\vec{\xi}(y)=\vec{\xi}^*(y)$ for all $y\in v$. 
Let $y_j$ be the point in $u^*$ with coordinates
\begin{equation*}
\xi_i^*(y_j) = \begin{cases} \xi_i(y_{j-1}) & \text{for all\ }i\neq k,\\
-\ell^*& \text{if\ }i=k\ \text{and}\ \ell^*<0,\\
2-\ell^*& \text{if\ }i=k\ \text{and}\ 1<\ell^*.
\end{cases}
\end{equation*}
Let $x$ be the point in $v$ with coordinates $\xi_i(x)= \xi_i(y_{j-1})$ for all $i\neq k$ and let $L:=L\union \{x\}$
\end{enumerate}
\end{enumerate} 
\end{alg}

There are many alternative ways to define random walks on $X$. 
For example, a fixed order of coordinates could be used at step $1$ of Algorithm~\ref{alg:rw} rather than randomly sampling the coordinate. 
Similarly, any distribution with zero mean and variance $t_0/m$ can be used to generate $\delta$ in step 2, provided the support of the distribution lies in $(-1,1)$. 
The proof of Theorem~\ref{thm:conv} applies directly in both these cases. 
The requirement that $\delta$ lies in $(-1,1)$ ensures that at most one codimension-1 cube is traversed so that step 3 of the algorithm makes sense. 
If this constraint is dropped, the proof of Theorem~\ref{thm:conv} still applies, but step~3 of the algorithm is more involved. 
It is also possible to sample the innovation uniformly from a small ball centred at $y_j$ at each iteration of the algorithm, which we return to in Section~\ref{sec:BHV}.  

We can use Algorithm~\ref{alg:rw} to simulate paths $\eta\in\paths{X}$ by linear interpolation between points. 
If at iteration $j$ step 3(a) was performed, then the corresponding section of path is the straight line segment joining $y_{j-1}$ and $y_j$ within an $n$-cube. 
Alternatively, if step 3(b) was performed at iteration $j$, then $L$ contains an associated point $x\in\Xcodim{1}$, and the corresponding path comprises the line segment $y_{j-1},x$ followed by the line segment $x,y_{j}$.  
The corresponding distribution on $\paths{X}$ is denoted $\walk{m}{X}$.

When $X=I^n$, Algorithm~\ref{alg:rw} corresponds to a reflected random walk on the $n$-cube. 
Standard theory gives weak convergence of random walk on $I^n$ to Brownian motion: 
\begin{equation*}
\walk{m}{I^n}\xrightarrow{w}\brownian{I^n}\quad \text{as}\  m\rightarrow\infty.
\end{equation*} 

The following lemma relates random walk on $X$ to random walk on $I^n$ via the projection map $\drop$. 

\begin{lem}\label{lem:reflect}
If $U\subset X$ is a Borel set then for any $k>0$, $\vec{f}\in F_k$ and $\vec{w}\in W(\vec{f})$, 
\begin{equation}\label{equ:pullback}
\walk{m}{X}\left(\widehat{C}_{\vec{f}}^{\vec{w}}(U)\right) = p(\vec{w})\walk{m}{I^n}\left(\drop(\widehat{C}_{\vec{f}}^{\vec{w}}(U))\right).
\end{equation}
\end{lem}

\begin{proof}
We first note that if $\eta$ is a path on $X$ generated by Algorithm~\ref{alg:rw} then $\drop(\eta)$ has the distribution $\walk{m}{I^n}$. 
This is because the coordinates $\vec{\xi}$ defined in the algorithm determine the image under $\drop$ but also precisely a random walk on $I^n$. 
This is equivalent to 
\begin{equation*}
\walk{m}{X}\left(\drop^{-1}(A)\right) = \walk{m}{I^n}(A) 
\end{equation*}
for all $A$ such that $\drop^{-1}(A)\in\borelX$. 
Given $\eta\in\drop^{-1}(A)$ where $A\subseteq C_{\vec{f}}$ for some sequence of faces $\vec{f}\in F_k$, then the distribution of $\phi(\eta)$ is determined by the probabilities $p(\vec{w})$. 
It follows that
\begin{equation*}
\walk{m}{X}\left( \drop^{-1}(A)\intersect \widehat{C}_{\vec{f}}^{\vec{w}} \right) = p(\vec{w})\walk{m}{I^n}(A). 
\end{equation*}
By substituting in $A = \drop(\widehat{C}_{\vec{f}}^{\vec{w}}(U))$ and applying Lemma~\ref{lem:symm}, we obtain equation~$\eqref{equ:pullback}$.
\end{proof}

\begin{defn}\label{def:trans_kernel}
Let $\W{x_0}{t_0}{m}$ denote the distribution on $X$ given by the end-points $y_r$ simulated by Algorithm~\ref{alg:rw}, or equivalently given by sampling a path $\eta$ from $\walk{m}{X}$ and taking $\eta(t_0)$. 
Let $\B{x_0}{t_0}$ be the distribution on $X$ induced in a similar way, but for which $\eta$ is drawn from $\brownian{X}$. 
Then
\begin{align*}
\W{x_0}{t_0}{m}(U) &= \walk{m}{X}(\widehat{C}(U))\\
\B{x_0}{t_0}(U) &= \brownian{X}(\widehat{C}(U))
\end{align*}
for any Borel set $U\subseteq X$, where $\widehat{C}(U)\subset\widehat{C}$ denotes the union of all sets of the form $\widehat{C}_{\vec{f}}^{\vec{w}}(U)$. 
\end{defn}

The distributions $\W{x_0}{t_0}{m}$ and $\B{x_0}{t_0}$ are the transition kernels for the $m$-step random walk and Brownian motion on $X$ respectively. 
More explicitly, 
\begin{equation*}
\B{x_0}{t_0}(U) = \sum_{k=0}^\infty \sum_{\vec{f}\in F_k}\sum_{\vec{w}\in W(\vec{f})} p(\vec{w}) \brownian{I^n}\left(\drop(\widehat{C}_{\vec{f}}^{\vec{w}}(U))\right). 
\end{equation*}
However, this equation is not useful computationally, since for any given $U\subseteq X$ the sum over walks which end in $U$ is not generally computationally feasible.  

If $P$ is a measure on a metric space, then a $P$-continuity set $U$ is a Borel set such that $P(\partial U)=0$ where $\partial U$ is the boundary of $U$. 
In order to prove weak convergence of $\W{x_0}{t_0}{m}$ to $\B{x_0}{t_0}$, it is sufficient to show that $\W{x_0}{t_0}{m}(U)\rightarrow\B{x_0}{t_0}(U)$ for any $\B{x_0}{t_0}$-continuity set $U$. 

\begin{lem}\label{lem:ctty}
If $U\subseteq X$ is a $\B{x_0}{t_0}$-continuity set, then 
\begin{enumerate}
\item $\widehat{C}_{\vec{f}}^{\vec{w}}(U)$ is a $\brownian{X}$-continuity set for any $\vec{f},\vec{w}$, and
\item $\drop( \widehat{C}_{\vec{f}}^{\vec{w}}(U) )$ is a $\brownian{I^n}$-continuity set.
\end{enumerate}
\end{lem}

\begin{proof}
Fix a $\B{x_0}{t_0}$-continuity set $U\subset X$ and $\vec{f}\in F_k,\vec{w}\in W(\vec{f})$. 
The boundary of $ \widehat{C}_{\vec{f}}^{\vec{w}}(U)$ decomposes into two pieces: $\widehat{C}_{\vec{f}}^{\vec{w}}(\partial U)$ and paths in the boundary of $\widehat{C}_{\vec{f}}^{\vec{w}}$.  
Both have zero measure with respect to $\brownian{X}$ so $\brownian{X}\left(\partial \widehat{C}_{\vec{f}}^{\vec{w}}(U)\right) =0$ and $\widehat{C}_{\vec{f}}^{\vec{w}}(U)$ is a $\brownian{X}$-continuity set.
Similarly, for fixed $\vec{f}$ and $\vec{w}$, $\partial\drop(\widehat{C}_{\vec{f}}^{\vec{w}}(U))$ decomposes into two parts: $\drop(\widehat{C}_{\vec{f}}^{\vec{w}}(\partial U))$ and a set of paths which lie the boundary of $C_{\vec{f}}$ and which therefore has zero measure with respect to $\brownian{I^n}$. 
Thus if $U$ is a $\B{x_0}{t_0}$-continuity set it follows that
\begin{align*}
0 &= \brownian{X}(\widehat{C}_{\vec{f}}^{\vec{w}}(\partial U))\\
&= p(\vec{w}) \brownian{I^n}(\drop(\widehat{C}_{\vec{f}}^{\vec{w}}(\partial U)))\ \text{using Definition~\ref{def:brownianX}}\\
&= p(\vec{w}) \brownian{I^n}(\partial\drop(\widehat{C}_{\vec{f}}^{\vec{w}}( U)))
\end{align*}
and so $\drop(\widehat{C}_{\vec{f}}^{\vec{w}}( U))$ is a continuity set with respect to $\brownian{I^n}$. 
\end{proof}

\begin{thm}\label{thm:conv}
\begin{equation}\label{equ:weakconv}
\W{x_0}{t_0}{m}\xrightarrow{w}\B{x_0}{t_0}
\end{equation}
as $m\rightarrow\infty$, where $\xrightarrow{w}$ denotes weak convergence of probability measures. 
\end{thm}

\begin{proof}
Fix any Borel set $U\subseteq X$, and fix $k\geq 0$, $\vec{f}\in F_k$ and $\vec{w}\in W(\vec{f})$. 
If $U$ is a $\B{x_0}{t_0}$-continuity set then equation~$\eqref{equ:pullback}$ gives
\begin{align}
\walk{m}{X}(\widehat{C}_{\vec{f}}^{\vec{w}}(U)) &= p(\vec{w})\walk{m}{I^n}(\drop(\widehat{C}_{\vec{f}}^{\vec{w}}(U)))\nonumber\\
& \rightarrow p(\vec{w})\brownian{I^n}(\drop(\widehat{C}_{\vec{f}}^{\vec{w}}(U)))\ \text{as\ }m\rightarrow\infty\label{equ:conv}\\
& = \brownian{X}(\widehat{C}_{\vec{f}}^{\vec{w}}(U)).\nonumber
\end{align}
Convergence in equation~$\eqref{equ:conv}$ occurs because random walk on $I^n$ converges weakly to Brownian motion, and by Lemma~\ref{lem:ctty}, $\drop(\widehat{C}_{\vec{f}}^{\vec{w}}(U))$ is a continuity set with respect to $\brownian{I^n}$.  

Now for any continuity set $U$
\begin{equation*}
\W{x_0}{t_0}{m}(U) = \sum_{k=0}^{\infty}\sum_{\vec{f}\in F_k}\sum_{\vec{w}\in W(\vec{f})}
\walk{m}{X}(\widehat{C}_{\vec{f}}^{\vec{w}}(U)).
\end{equation*}
Since this is bounded above by $1$, for any $\epsilon>0$ there exists $K$ such that:
\begin{equation*}
|\sum_{k>K}\sum_{\vec{f}\in F_k}\sum_{\vec{w}\in W(\vec{f})}
\walk{m}{X}(\widehat{C}_{\vec{f}}^{\vec{w}}(U))|<\frac{\epsilon}{4}.
\end{equation*}
and 
\begin{equation*}
|\sum_{k>K}\sum_{\vec{f}\in F_k}\sum_{\vec{w}\in W(\vec{f})}
\brownian{X}(\widehat{C}_{\vec{f}}^{\vec{w}}(U))|<\frac{\epsilon}{4}
\end{equation*}
Then, by taking $m$ sufficiently large
\begin{equation*}
| \sum_{k\leq K}\sum_{\vec{f}\in F_k}\sum_{\vec{w}\in W(\vec{f})} \left( 
\walk{m}{X}(\widehat{C}_{\vec{f}}^{\vec{w}}(U)) - 
\brownian{X}(\widehat{C}_{\vec{f}}^{\vec{w}}(U)) 
 \right)  |<\frac{\epsilon}{2}
\end{equation*}
and so $|\W{x_0}{t_0}{m}(U)-\B{x_0}{t_0}(U)|<\epsilon$.
This proves the weak convergence in Equation~$\eqref{equ:weakconv}$.
\end{proof}

\section{Application: evolutionary trees}\label{sec:BHV}

The main application of Theorem~\ref{thm:conv} lies in the analysis of evolutionary trees. 
A number of different spaces of evolutionary trees have been constructed, and the most studied space, due to Billera, Holmes and Vogtmann~\cite{bill01} and known as BHV tree space, is a cubical complex. 
The BHV tree space $\treespace{N}$ consists of trees with $N$ leaves labelled $1,\ldots,N$. 
The trees are edge weighted: each edge has an associated weight, or length, which takes values in $\R_{>0}=\{x\in \R: x>0\}$. 
Each tree contains a unique vertex of degree two, called the \emph{root}, which is labelled $0$. 
All vertices other than the leaves and the root have degree strictly greater than two, and are called \emph{internal vertices}. 
If all internal vertices have degree three, then the tree is called \emph{resolved} and it contains $2N-2$ edges. 
If one or more internal vertices has degree four or higher then the tree is called \emph{unresolved} and it contains $<2N-2$ edges. 
BHV tree space $\treespace{N}$ parametrizes all such resolved and unresolved trees. 

Cutting any edge on a tree induces a bipartition of the labels $\{0,1,\ldots,N\}$ into two disjoint subsets. 
Any such bipartition is called a \emph{split}, and the set of splits associated with tree is called its \emph{topology}. 
Every edge is identified with a split, so the terms are used interchangeably. 
Arbitrary collections of splits do not generally represent tree topologies: a certain compatibility condition on the collection of splits must be satisfied for this to be true. 
The splits which end in leaves are displayed by all trees, and are called \emph{pendant} splits. 
It follows that
\begin{equation*}
\treespace{N} = \R_{> 0}^N\times\treespaceint{N}
\end{equation*}
where the first term in the product represents the lengths of the pendant splits and the second term parametrizes the lengths of internal splits (namely those corresponding to edges which do not contain a leaf). 
If we ignore pendant splits then the set of trees with a given fully resolved topology is parametrized by $\R_{>0}^{N-2}$, where each axis corresponds to the length of one of the $N-2$ internal splits in the topology. 
The boundary of $\R_{>0}^{N-2}$ corresponds to trees for which one or more splits have been shrunk down to length zero and removed from the tree topology, and so corresponds to unresolved trees. 
The copy of $\R_{\geq 0}^{N-2}$ obtained in this way is called an \emph{orthant}. 
There are $M=(2N-3)!!$ different possible fully resolved topologies, so $\treespaceint{N}$ consists of $M$ copies of $\R_{\geq 0}^{N-2}$ which are glued along their boundaries. 
The gluing identifies unresolved trees which can be obtained by removing splits from distinct tree topologies. 

By filling each orthant with an infinite array of $(N-2)$-cubes, it can be seen that $\treespaceint{N}$ is a cubical complex. 
Billera, Holmes and Vogtmann~\cite{bill01} showed that the local Euclidean metric on each cube extends to a globally defined metric, and that $\treespace{N}$ is a so-called non-positively curved space~\cite{BH1999}. 
Brownian motion on $\treespaceint{N}$ depends fundamentally on how the codimension-$1$ faces of orthants are glued together, but not on higher codimension faces. 
Any internal edge $e$ in a fully resolved rooted tree gives rise to a subtree $((A,B),C)$ where $A,B$ are the two subtrees descended from $e$ and $C$ is the subtree attached to the vertex $v$ of $e$ closest to the root. 
Shrinking $e$ down to zero length and removing it results in $v$ having degree $4$, and the subtree descended from $v$ is then $(A,B,C)$. 
A single edge can be added to this tree to give either $((B,C),A)$ or $((C,A),B)$ as the subtree descended from $v$. 
It follows that each codimension-$1$ face of an orthant $\R_{\geq 0}^{N-2}$ is shared by two other such orthants. 
The bipartite graph representing connections between orthants in $\treespaceint{N}$ therefore contains $(2N-3)!!$ vertices representing the orthants, each with degree $N-2$, and $\frac{1}{3}(N-2)\times(2N-2)!!$
vertices representing codimension-$1$ faces, each with degree 3. 
The stochastic processes we have defined on cubical complexes extend trivially from $\treespaceint{N}$ to the whole of $\treespace{N}$ by imposing a reflecting boundary when each pendant edge length is zero. 

Theorem~\ref{thm:conv} establishes that while distributions of the form $\B{x_0}{t_0}$ can be used to construct parametric statistical models on $\treespace{N}$, computational inference can be performed using the corresponding distributions $\W{x_0}{t_0}{m}$ and Algorithm~\ref{alg:rw}. 
In fact, Algorithm~\ref{alg:rw} can be simplified to operate more explicitly on the set of orthants, rather than on cubes, in the following way. 
As input the algorithm takes a fully resolved tree $x_0\in \treespace{N}$, a value $t_0>0$ and the number of steps $m$ of the random walk. 
As previously, we fix $\epsilon = (3t_0/m)^{1/2}$. 

\begin{alg}\label{alg:rw-trees}
Let $y_0=x_0$. 
For $j=1,2,\ldots,r$ where $r=m\times (N-2)$ repeat the following:
\begin{enumerate}
\item Pick an internal split $e$ from ${y_{j-1}}$ uniformly at random. 
\item Sample a realization $\delta$ from $U(-\epsilon,\epsilon)$ and let $\ell^\ast=\delta_j+\edgelen{y_{j-1}}{e}$ where $\edgelen{y_{j-1}}{e}$ is length of $e$ in $y_{j-1}$.
\item Set $y_j:=y_{j-1}$ and then change a single edge length in $y_i$ as follows:
\begin{enumerate}
    \item If $\ell^\ast>0$ set $\edgelen{y_j}{e}:=\ell^\ast$.
    \item Otherwise let $e^\ast$ be a split selected uniformly at random from the set $\{e,e',e''\}$ where $e',e''$ are the two alternative splits which can replace $e$ in $y_{j-1}$.  
    Replace $e$ with $e^\ast$ in $y_j$ and set $\edgelen{y_j}{e^\ast}:=-\ell^\ast$.
    \end{enumerate}
\end{enumerate} 
\end{alg}

This algorithm samples from the same distribution as Algorithm~\ref{alg:rw} when $\treespace{N}$ is regarded as a cubical complex, but it is simpler to implement and has the advantage of operating more transparently on trees. 

Two spaces in the literature which are related to BHV tree space, are also cubical complexes to which Theorem~\ref{thm:conv} applies. 
The first is a space of so-called \emph{ultrametric} phylogenetic trees \cite{gav2016}. 
Ultrametric phylogenetic trees have edge lengths proportional to the length of time between speciation events represented by each vertex. 
This induces a constraint on the trees, so that the path length of each leaf from the root is the same for all leaves, each of which represents a present-day species. 
Secondly, a space of phylogenetic networks has recently been constructed \cite{dev16} which is a cubical complex. 
Phylogenetic networks are a generalization of phylogenetic trees that model `non-vertical' patterns of evolution such as hybridization events.  

Other spaces of trees exist which are polyhedral complexes in which the cells are not cubes. 
The `projective tree space' Zairis et al.~\cite{zair16}, for example, consists of the subset of BHV tree space for which the sum of edge lengths on each tree is $1$.
This space is a simplicial complex in which every cell of maximum dimension is isometric to the standard Euclidean simplex. 
While Theorem~\ref{thm:conv} does not apply directly to this space, it is evident that methods used to prove the theorem can be adapted to apply to polyhedral complexes in which every cell is isometric to some given fixed polyhedron. 
More specifically, suppose $X$ is a polyhedral complex for which the highest dimension of any cell is $n$, and suppose that all $n$-cells in $X$ are isometric to some given polyhedron $S$. 
Furthermore suppose the following condition holds: for any two $n$-polyhedra $S_1,S_2$ in $X$ which are glued together via a shared codimension-$1$ face, suppose the isometry which defines the gluing extends to determine a unique isometry between $S_1$ and $S_2$. 
It follows that a projection map $\drop$ from paths on $X$ to paths on $S$ can be defined, and most other elements of the proof of Theorem~\ref{thm:conv} can be extended to this new setting. 
Some care, however, is required to adapt Algorithm~\ref{alg:rw}. 
In the new setting, the innovation takes the following form. 
Suppose $y_{j-1}$ lies in the interior of an $n$-polyhedron $S$.  
Let $\delta$ be a sample from an isotropic multivariate normal distribution with variance proportional to $t_0/m$. 
If $y_{j-1}+\delta$ lies in $S$ then that point is taken to define $y_j$. 
However, if $y_{j-1}+\delta$ lies outside $S$ then $y_j$ is obtained by following a path of length $\|\delta\|$ through $X$ from $y_{j-1}$ which potentially moves between different $n$-polyhedra. 
The path is initially in direction $\delta$, but it can `reflect' at codimension-$1$ boundaries of polyhedra. 
Crucially, at each codimension-$1$ boundary, it moves with equal probability into each adjacent $n$-polyhedron. 
Figure~\ref{fig:rw-general} illustrates the innovation step of the random walk. 
Using this more general random walk algorithm, the proof of Theorem~\ref{thm:conv} generalizes to this wider class of polyhedral complex, and so, for example, random walks can be used to approximate Brownian motion on the projective tree space of Zairis et al. 

\begin{figure}
\begin{center}
\includegraphics[scale=0.9]{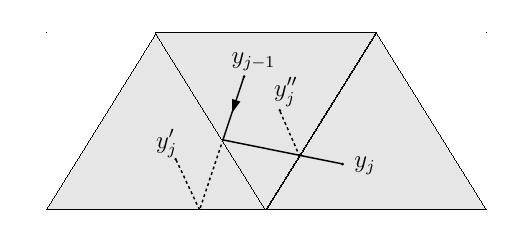}
\caption{One step of a random walk for a simplicial complex $X$ consisting of three equilateral triangles. 
Starting at $y_{j-1}$, the random vector $\delta$ determines the initial direction of the path, as indicated by the arrow. 
At the first codimension-$1$ boundary, the path reflects with probability $1/2$. 
The path continues across the second boundary, again with probability $1/2$ to give the next point $y_j$. 
The dotted lines show the alternatives as each boundary is crossed. 
Point $y_j'$ is reached with probability $1/2$, whereas point $y_j''$ has probability $1/4$. 
The length of each of the three possible paths is $\|\delta\|$.}
\label{fig:rw-general}
\end{center}
\end{figure}

\section{Conclusion}\label{sec:concl}

Theorem~\ref{thm:conv} establishes that random walks on certain cubical complexes converge to a well-defined limit, namely a Brownian motion on the complex. 
The motivation for proving the theorem is the idea that transition kernels of Markov processes on cubical complexes are a useful source of parametric distributions which can be used to construct statistical models. 
Inference procedures for models on BHV tree space constructed in this way are being developed by the author. 
The theorem is limited in a number of ways. 
First, it is important to note that we have not proved a version of Donsker's theorem on a cubical complex $X$; in particular we have not proved convergence of the distributions of sample paths $\walk{m}{X}$ to $\brownian{X}$ as $m\rightarrow\infty$. 
Instead, the theorem is concerned with convergence of the transition kernels of the Markov processes on $X$. 
Secondly, the theorem is limited by the assumption that $x_0$, the initial point of the Brownian motion, lies in the interior of some cube of dimension $n$ which is maximal in $X$. 
In fact, provided $x_0$ does not lie in a codimension-$2$ face of an $n$-cube or a face of higher codimension, the projection map $\drop$ can still be defined and the proof of Theorem~\ref{thm:conv} goes through, subject to a minor change in the definition of Algorithm~\ref{alg:rw} to take account of the position of $x_0$. 
When $x_0$ lies in a codimension-$2$ face of an $n$-cube, there is no canonical definition of $\drop$. 
%A general cubical complex $X$ can have non-trivial holonomy, so there is generally no continuous map $X\rightarrow I^n$ which acts as an isometry between each $n$-cube and $I^n$. 
%The projection map $\drop$  operates specifically on \emph{paths} in $X$, and so the existence of $\drop$ when $x_0$ lies in the interior of an $n$-cube does not contradict the previous statement.
We have also ignored the case that $x_0$ is contained in a cube of dimension $n'<n$ which is \emph{not} a face of an $n$-cube, where $n$ is the maximal dimension. 
Under these circumstances, the particle will diffuse until it hits a cube of higher dimension and will almost surely not return to cubes of dimension $n'$ or lower. 

% More general Markov processes. More general spaces. 
The Brownian motion studied in this article could be generalized in several different ways by analogy with diffusion processes in Euclidean space. 
Non-positively curved cubical complexes, which include BHV tree space, are an important class of complex. 
In a non-positively curved cubical complex, every pair of points is connected by a unique geodesic, or path of shortest length \cite{BH1999}. 
Geodesics enable random walk steps to be combined with retraction towards a fixed point, thereby yielding a mean-reverting stochastic process analogous to an Ornstein-Uhlenbeck process. 
Alternatively, parallel translation across codimension-$1$ faces of cubes could enable a non-trivial covariance structure to be used, producing a diffusion process with `preferred directions'. 
Both types of stochastic process would require theorems analogous to our convergence result in order to allow approximation by analogous discrete-time processes. 
Instead of considering more general stochastic processes on cubical complexes, a second direction is to consider random walks and Brownian motion on a wider class of complexes such as a general Euclidean simplicial complex. 
As mentioned in Section~\ref{sec:BHV}, our results extend readily to a certain highly symmetric class of polyhedral complex. 
A proof for a general simplicial complex appears more challenging, due to the lack the symmetry properties used extensively in the proof of Theorem~\ref{thm:conv}.

%--------------------------------------------------------------------------------

%--------------------------------------------------------------------------------


\begin{thebibliography}{10}

\bibitem{fer14b}
A.~Feragen, S.~Huckemann, J.~S. Marron, E.~Miller, Asymptotic statistics on
  stratified spaces, Oberwolfach Reports 44.

\bibitem{fer14a}
A.~Feragen, M.~Nielsen, E.~B.~V. Jensen, A.~du~Plessis, F.~Lauze, Geometry and
  statistics: Manifolds and stratified spaces, J. Math. Imaging Vis. 50 (2014)
  1--4.

\bibitem{ooda14}
J.~S. Marron, A.~M. Alonso, Overview of object oriented data analysis,
  Biometrical J. 56~(5) (2014) 732--753.

\bibitem{kendall1999}
D.~G. Kendall, D.~Barden, T.~K. Carne, H.~Le, Shape and shape theory, John
  Wiley \& Sons, 1999.

\bibitem{bill01}
L.~Billera, S.~Holmes, K.~Vogtmann, Geometry of the space of phylogenetic
  trees, Adv. Appl. Math. 27 (2001) 733--767.

\bibitem{fera11a}
A.~Feragen, F.~Lauze, P.~Lo, M.~de~Bruijne, M.~Nielsen, Geometries on spaces of
  treelike shapes, in: Computer Vision--ACCV 2010, Springer, 2011, pp.
  160--173.

\bibitem{gav2016}
A.~Gavryushkin, A.~J. Drummond, The space of ultrametric phylogenetic trees, J.
  Theor. Bio. 403 (2016) 197--208.

\bibitem{dev16}
S.~Devadoss, S.~Petti, A space of phylogenetic networks, arXiv preprint
  arXiv:1607.06978.

\bibitem{gromov1987}
M.~Gromov, Hyperbolic groups, in: S.~M. Gersten (Ed.), Essays in group theory,
  Vol.~8 of Mathematical Sciences Research Institute Publications, Springer,
  1987, pp. 75--263.

\bibitem{niblo1998}
G.~A. Niblo, L.~D. Reeves, The geometry of cube complexes and the complexity of
  their fundamental groups, Topology 37~(3) (1998) 621--633.

\bibitem{mill15}
E.~Miller, M.~Owen, J.~S. Provan, Polyhedral computational geometry for
  averaging metric phylogenetic trees, Adv. Appl. Math. 68 (2015) 51--91.

\bibitem{nye11}
T.~Nye, Principal components analysis in the space of phylogenetic trees, Ann.
  Stat. 39 (2011) 2716--2739.

\bibitem{willis16}
A.~Willis, Confidence sets for phylogenetic trees, arXiv preprint
  arXiv:1607.08288.

\bibitem{barden2013central}
D.~Barden, H.~Le, M.~Owen, et~al., Central limit theorems for {F}r{\'e}chet
  means in the space of phylogenetic trees, Electron. J. Probab 18~(25) (2013)
  1--25.

\bibitem{hotz13}
T.~Hotz, S.~Huckemann, H.~Le, J.~Marron, J.~Mattingly, E.~Miller, J.~Nolen,
  M.~Owen, V.~Patrangenaru, S.~Skwerer, Sticky central limit theorems on open
  books, Ann. Appl. Probab. 23~(6) (2013) 2238--2258.

\bibitem{barden2017}
D.~Barden, H.~Le, The logarithm map, its limits and {F}r{\'e}chet means in
  orthant spaces, arXiv preprint arXiv:1703.07081.

\bibitem{wey16}
G.~Weyenberg, R.~Yoshida, D.~Howe, Normalizing kernels in the
  Billera-Holmes-Vogtmann treespace, IEEE ACM T. Comput. Bi. (2016)
  doi:10.1109/TCBB.2016.2565475.

\bibitem{nye14a}
T.~Nye, M.~White, Diffusion on some simple stratified spaces, J. Math. Imaging
  Vis. 50 (2014) 115--125.

\bibitem{frank1984}
D.~H. Frank, S.~Durham, Random motion on binary trees, J. Appl. Probab. 21~(1)
  (1984) 58--69.

\bibitem{brin01}
M.~Brin, Y.~Kifer, Brownian motion, harmonic functions and hyperbolicity for
  {E}uclidean complexes, Math. Z. 237~(3) (2001) 421--468.

\bibitem{enri01}
N.~Enriquez, Y.~Kifer, Markov chains on graphs and {B}rownian motion, J. Theor.
  Probab. 14~(2) (2001) 495--510.

\bibitem{kost2012}
V.~Kostrykin, J.~Potthoff, R.~Schrader, Brownian motion on metric graphs,
  J. Math. Phys. 53 (2012) 095205.

\bibitem{fernos2017}
T.~Fern{\'o}s, The {F}urstenberg--{P}oisson boundary and {CAT(0)} cube
  complexes, Ergod. Theor. Dyn. Syst. (2017) 1--44.

\bibitem{kendall1977}
D.~G. Kendall, The diffusion of shape, Adv. Appl. Probab. 9~(3) (1977)
  428--430.

\bibitem{BH1999}
M.~R. Bridson, A.~Haefliger, Metric Spaces of Non-Positive Curvature, Vol. 319,
  Springer-Verlag, 2011.

\bibitem{zair16}
S.~Zairis, H.~Khiabanian, A.~J. Blumberg, R.~Rabadan, Genomic data analysis in
  tree spaces, arXiv preprint arXiv:1607.07503.

\end{thebibliography}
\end{document}